\documentclass{article}

\usepackage[a4paper]{geometry}
\usepackage{subcaption}
\usepackage{graphicx}
\usepackage{xspace}
\usepackage{lineno}
\usepackage{wrapfig}
\usepackage{paralist} 
\usepackage{amsthm}
\usepackage{amsmath}
\usepackage{amssymb}

\newtheorem{theorem}{Theorem}
\newtheorem{lemma}{Lemma}
\newtheorem{corollary}{Corollary}

\title{Polyline Drawings with Topological Constraints}

\author{
	Emilio Di Giacomo\thanks{Universit\`a degli Studi di Perugia, Perugia, Italy,\newline 
		\indent \texttt{\{emilio.digiacomo, giuseppe.liotta, fabrizio.montecchiani\}@unipg.it}} 
	\and Peter Eades\thanks{University of Sydney, Sydney, Australia, \texttt{peter.d.eades@gmail.com}} 
	\and Giuseppe Liotta\footnotemark[1] 
	\and Henk Meijer\thanks{University College Roosevelt, Middelburg, The Netherlands, \texttt{h.meijer@ucr.nl}} 
	\and Fabrizio Montecchiani\footnotemark[1]
}

\date{}

\begin{document}

\maketitle

\begin{abstract}
Let $G$ be a simple topological graph and let $\Gamma$ be a polyline drawing of $G$.  We say that  $\Gamma$ \emph{partially preserves the topology} of $G$ if it has the same external boundary, the same rotation system, and the same set of crossings as $G$. Drawing $\Gamma$ \emph{fully preserves the topology} of $G$  if the planarization of $G$ and the planarization of $\Gamma$ have the same planar embedding. We show that if the set of crossing-free edges of $G$ forms a connected spanning subgraph, then $G$ admits a polyline drawing that partially preserves its topology and that has curve complexity at most three (i.e., at most three bends per edge). If, however, the set of crossing-free edges of $G$ is not a connected spanning subgraph, the curve complexity may be $\Omega(\sqrt{n})$. Concerning drawings that fully preserve the topology, we show that if $G$ has skewness $k$, it admits one such drawing with curve complexity at most $2k$; for skewness-1 graphs, the curve complexity can be reduced to one, which is a tight bound. We also consider optimal $2$-plane graphs and discuss trade-offs between curve complexity and crossing angle resolution of drawings that fully preserve the topology.
\end{abstract}

\section{Introduction}
 A fundamental result in graph drawing is the so-called ``stretchability theorem''~\cite{fary48,Stein51,Wagner36}: Every planar simple topological graph admits a straight-line drawing that preserves its topology. One may ask whether a similar theorem holds for non-planar simple topological graphs. Motivated by the fact that a straight-line drawing  may  not be possible even for a planar graph plus an edge~\cite{eades2015}, we allow bends along the edges and measure the quality of the computed drawings in terms of their \emph{curve complexity}, defined as the maximum number of bends per edge.

 Let $G$ be a simple topological graph and let $\Gamma$ be a polyline drawing of $G$. (Note that, by definition of simple topological graph, $G$ has neither multiple edges nor self-loops; see also Section~\ref{se:preli} for formal definitions.)  Drawing $\Gamma$  \emph{fully preserves the topology} of $G$ if the planarization of $G$ (i.e., the planar simple topological graph obtained from $G$ by replacing crossings with dummy vertices) and the planarization of $\Gamma$ have the same planar embedding. Eppstein et al.~\cite{DBLP:journals/corr/EppsteinGSU16} prove the existence of a simple arrangement of $n$ pseudolines that, when drawn with polylines, it requires at least one pseudoline to have $\Omega(n)$ bends. It is not hard to see that the result by Eppstein et al. implies the existence of an $n$-vertex simple topological graph such that any polyline drawing that fully preserves its topology has curve complexity $\Omega(n)$ (see Corollary~\ref{co:eppstein} in Section~\ref{se:preli}). This lower bound naturally suggests two research directions: (i) ``Trade'' curve complexity for accuracy in the preservation of the topology and (ii) Describe  families of simple topological graphs for which polyline drawings that fully preserve their topologies and that have low curve complexity can be computed.
 
 \begin{figure}[t]
	\centering
	\begin{minipage}[b]{.18\textwidth}
		\centering 
		\includegraphics[page=1, width=\textwidth]{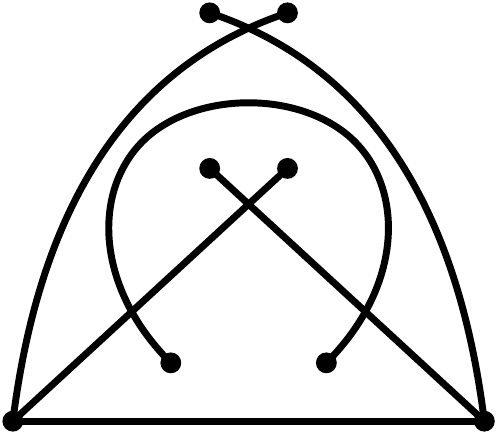}
		\subcaption{~}\label{fi:intro-a}
	\end{minipage}
	\hfil
	\begin{minipage}[b]{.18\textwidth}
		\centering 
		\includegraphics[page=2, width=\textwidth]{figure/intro-example}
		\subcaption{~}\label{fi:intro-b}
	\end{minipage}
	\caption{(a) A topological graph $G$ that requires at least $1$ bend in any polyline drawing that fully preserves its topology. (b) A straight-line drawing that partially preserves the topology~of~$G$.}
\end{figure}

 Concerning the first research direction, we consider the following relaxation of topology preserving drawing. A polyline drawing of a simple topological graph $G$  \emph{partially preserves the topology} of $G$ if it has the same rotation system, the same external boundary, and the same set of crossings as $G$, while it may not preserve the order of the crossings along an edge. It may be worth recalling that some (weaker) notions of topological equivalence between  graphs have been already considered in the literature. For example, Kyn\v{c}l ~\cite{Kyncl2011,KYNCL20091676} and Aichholzer et al.~\cite{aafhpprsv-agdsc-15,ahpsv-dmgdc-15} study \emph{weakly isomorphic} simple topological graphs: Two simple topological graphs are \emph{weakly isomorphic} if they have the same set of vertices, the same set of edges, and the same set of edge crossings. Note that a drawing $\Gamma$ that partially preserves the topology of a simple toplogical graph $G$ is weakly isomorphic to $G$ and, in addition, it has the same rotation system and the same external boundary as $G$. Also, Kratochv\'{\i}l, Lubiw, and Ne\v{s}et\v{r}il~\cite{kln-nstl-91} define the notion of  \emph{abstract topological graph} as a pair $(G,\chi)$, where $G$ is a graph and $\chi$ is a set of pairs of crossing edges; a \emph{strong realization} of $G$ is a drawing $\Gamma$ of $G$ such that two edges of $\Gamma$ cross if and only if they belong to $\chi$. The problem of computing a drawing that partially preserves a topology may be rephrased as the problem of computing a strong realization of an abstract topological graph for which a rotation system and an external boundary are given in input. A different relaxation of the topology preservation is studied by Durocher and Mondal, who proved bounds on the curve complexity of drawings that preserve the thickness of the input graph~\cite{durocher2016}. 
 
 Concerning the second research direction, we investigate the curve complexity of polyline drawings that fully preserve the topology of meaningful families of \emph{beyond-planar} graphs, that are families of non-planar graphs for which some crossing configurations are forbidden (see, e.g.,~\cite{DBLP:journals/jgaa/BekosKM18,DBLP:journals/corr/abs-1804-07257} for surveys and special issues on beyond-planar graph drawing). In particular, we focus on graphs with skewness $k$, i.e., non-planar graphs that can be made planar by removing at most $k$ edges, and on $2$-plane graphs, i.e., non-planar graphs for which any edge is crossed at most twice. Note that a  characterization of those graphs with skewness one having a straight-line drawing that fully preserves the topology is presented in~\cite{eades2015}. Also, all $1$-plane graphs (every edge can be crossed at most once) admit a polyline drawing with curve complexity one that fully preserves the topology and such that any crossing angle is $\frac{\pi}{2}$~\cite{DBLP:conf/ewcg/Chaplick18}.

 Our results can be listed as follows. Let $G$ be a simple topological graph.

\begin{itemize}
\item If the subgraph of $G$ formed by the uncrossed edges and all vertices of $G$, called {\em planar skeleton}, is connected, then $G$ admits a polyline drawing with curve complexity three that partially preserves its topology. If the planar skeleton is biconnected the curve complexity can be reduced to one, which is worst-case optimal (Section~\ref{se:partially}).

\item For the case that the planar skeleton of $G$ is not connected, we prove that the  curve complexity may be $\Omega(\sqrt{n})$ (Section~\ref{se:partially}).

\item If $G$ has skewness $k$, then $G$ admits a polyline drawing with curve complexity $2k$ that fully preserves its topology. When $k=1$, the curve complexity can be reduced to one, which is worst-case optimal (Section~\ref{se:fully}).

\item If $G$ is optimal 2-plane (i.e., it is $2$-plane and it has $5n - 10$ edges), then $G$ admits  a drawing that  fully preserves its topology and with two bends in total, and a drawing  that  fully preserves its topology, with at most two bends per edge, and with optimal crossing angle resolution. The number of bends per edge can be reduced to one while maintaining the crossing angles arbitrarily close to $\frac{\pi}{2}$ (Section~\ref{se:fully}).
\end{itemize}

We conclude the introduction with an example about the difference between a drawing that fully preserves  and one that partially preserves a given topology. Figure~\ref{fi:intro-a} shows a simple topological graph for which every polyline drawing fully preserving its topology has at least one bend on some edge.  Figure~\ref{fi:intro-b} shows a  drawing of the same graph that partially preserve its topology and has no bends.


\section{Preliminaries}\label{se:preli}

A \emph{simple topological graph} is a drawing of a  graph in the plane such that: (i) vertices are distinct points, (ii) edges are Jordan arcs that connect their endvertices and  do not pass through other vertices, (iii) any two edges intersect at most once by either making a proper crossing or by sharing a common endvertex, and (iv) no three edges pass through the same crossing. A simple topological graph has neither multiple edges (otherwise there would be two edges intersecting twice), nor self-loops (because the endpoints of a Jordan arc do not coincide). A simple topological graph is \emph{planar} if no two of its edges cross. A planar simple topological graph $G$ partitions the plane into topological connected regions, called \emph{faces} of $G$. The unbounded face is called the \emph{external face}.
The \emph{planar embedding} of a simple planar topological graph $G$ fixes the \emph{rotation system} of $G$, defined as the clockwise circular order of the edges around each vertex, and the external face of $G$. The \emph{planar skeleton} of a simple topological graph $G$ is the subgraph of $G$ that contains all vertices  and only the uncrossed edges of $G$. A simple topological graph obtained from $G$ by adding uncrossed edges (possibly none) is called a \emph{planar augmentation} of $G$.

Let $\mathcal{L}$ be an arrangement of $n$ pseudolines; a \emph{polyline realization} $\Gamma_\mathcal{L}$ of $\mathcal{L}$ represents each pseudoline as a polygonal chain while preserving the topology of $\mathcal{L}$. The \emph{curve complexity of $\Gamma_\mathcal{L}$} is the maximum number of bends per  pseudoline in $\Gamma_\mathcal{L}$. The \emph{curve complexity of $\mathcal{L}$} is the minimum curve complexity over all polyline realizations of  $\mathcal{L}$. The \emph{graph associated with $\mathcal{L}$} is a simple topological graph $G_\mathcal{L}$ defined as follows. Let $C$ be a circle of sufficiently large radius such that all crossings of $\mathcal{L}$ are inside $C$ and every pseudoline intersects the boundary of $C$ exactly twice. Replace each crossing between $C$ and a pseudoline with a vertex, remove the portions of each pseudoline that are outside $C$, add an apex vertex $v$ outside $C$, and connect $v$ to the vertices of $C$ with crossing-free edges. See Fig.~\ref{fi:arrangement} for an example.

\begin{figure}[t]
	\centering
	\begin{minipage}[b]{.28\textwidth}
		\centering 
		\includegraphics[page=1, width=\textwidth]{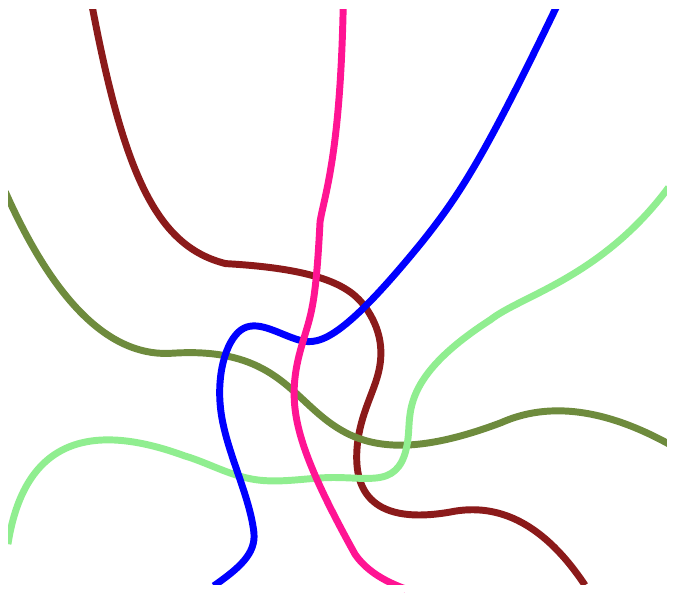}
		\subcaption{~}\label{fi:arrangement-a}
	\end{minipage}
	\hfil
	\begin{minipage}[b]{.28\textwidth}
		\centering 
		\includegraphics[page=2, width=\textwidth]{figure/arrangement}
		\subcaption{~}\label{fi:arrangement-b}
	\end{minipage}
	\caption{(a) An arrangement of pseudolines $\mathcal{L}$. (b) The graph $G_{\mathcal{L}}$ associated with $\mathcal{L}$. }
\label{fi:arrangement}
\end{figure}

\begin{lemma}\label{le:eppstein}
Let $\mathcal{L}$ be an arrangement of $n$ pseudolines and let $G_\mathcal{L}$ be the simple topological graph associated with $\mathcal{L}$. Every polyline drawing of $G_\mathcal{L}$ that fully preserves its topology has curve complexity $\Omega(f(n))$ if and only if $\mathcal{L}$ has curve complexity $\Omega(f(n))$.
\end{lemma}
\begin{proof}
	Assume that every polyline drawing of $G_\mathcal{L}$ that fully preserves its topology has curve complexity $\Omega(f(n))$ and suppose, as a contradiction, that $\mathcal{L}$ has a polyline representation $\Gamma_{\mathcal{L}}$ with $o(f(n))$ bends. We draw a circle $C$ on $\Gamma_{\mathcal{L}}$ so that all crossings and bends are inside $C$. We place a vertex at each crossing between $C$ and a pseudoline of $\mathcal{L}$ and remove the portions of each pseudolines that are outside $C$. We obtain a drawing of $G_\mathcal{L}$ except for the apex vertex $v$ and its incident edges. We place $v$ outside $C$ sufficiently far so that it is possible to connect it to all the other vertices by drawing each edge with at most $1$ bend. The resulting drawing is a drawing of $G_\mathcal{L}$ with curve complexity $o(f(n))$, a contradiction.
	
	\begin{figure}[h]
		\centering
		\begin{minipage}[b]{.28\textwidth}
			\centering 
			\includegraphics[page=1, width=\textwidth]{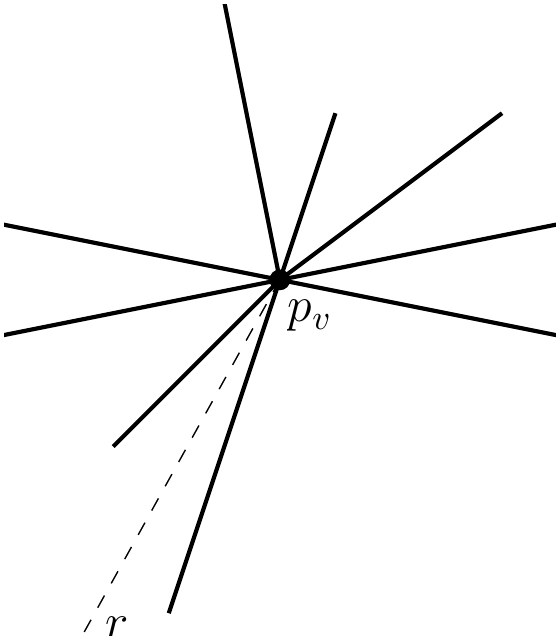}
			\subcaption{~}\label{fi:center-a}
		\end{minipage}
		\hfil
		\begin{minipage}[b]{.28\textwidth}
			\centering 
			\includegraphics[page=2, width=\textwidth]{figure/center}
			\subcaption{~}\label{fi:center-b}
		\end{minipage}
		\caption{Illustration for the proof of Lemma~\ref{le:eppstein}. (a) Point $p_v$ with all incident edges. (b) Separation of the different edges and extension to infinity.\label{fi:center}}
	\end{figure}
	
	Assume now that every polyline representation of $\mathcal{L}$ has curve complexity $\Omega(f(n))$ and suppose that $G_\mathcal{L}$  admits a polyline drawing $\Gamma$ whose curve complexity is $b \in o(f(n))$. For each pseudoline $\ell \in \mathcal{L}$ we have two vertices $u_1$ and $u_2$ on $C$ and therefore three edges in $G_\mathcal{L}$: $(u_1,u_2)$, $(v,u_1)$, and $(v,u_2)$, where $v$ is the apex vertex of $G_\mathcal{L}$. Each of these three edges has at most $b$ bends and therefore their union is a closed curve $\gamma_{\ell}$ with at most $3b$ bends. Let $p_v$ be the point representing $v$ in $\Gamma$. Suppose first that there exists a half-line $r$ with origin $p_v$ that does not intersect $\Gamma$ except at $p_v$. It is possible to choose a set of $2n$ lines parallel to $r$ and sufficiently close to it so that each curve $\gamma_{\ell}$ can be cut in a neighborhood of $p_v$ and extended to infinity by using two of the parallel lines (see Fig.~\ref{fi:center}). The resulting drawing is a polyline realization of $\mathcal{L}$ with curve complexity $3b+2 \in o(f(n))$, a contradiction. If the half-line $r$ does not exist, starting from $p_v$ and following the boundary of the external face, we can draw a polyline with at most $3b$ bends that reaches a point $p'_v$ for which the half-line $r$ exists and use $p'_v$ to extend to infinity the polylines representing the pseudolines. The final drawing has curve complexity $5b+\Theta(1) \in o(f(n))$, again a contradiction.
\end{proof}

Lemma~\ref{le:eppstein} and the result of Eppstein et al.~\cite{DBLP:journals/corr/EppsteinGSU16} proving the existence of an arrangement of $n$ pseudolines with curve complexity $\Omega(n)$ imply the following.

\begin{corollary}\label{co:eppstein}
There exists a simple topological graph with $n$ vertices such that any drawing that fully preserves its topology has curve complexity $\Omega(n)$.
\end{corollary}

In the next section we study a relaxation of the concept of topology preservation by which we derive constant upper bounds on the curve complexity.

\section{Polyline Drawings that Partially Preserve the Topology}\label{se:partially}

A polygon $P$ is \emph{star-shaped} if there exists a set of points, called the \emph{kernel} of $P$, such that for every point $z$ in this set and for each point $p$ of on the boundary of $P$, the segment $\overline{zp}$ lies entirely within $P$. A simple topological graph is \emph{outer} if all its vertices are on the external boundary and all the edges of the external boundary are uncrossed. Let $G$ be an outer simple topological graph with $n \ge 3$ vertices and let $P$ be a star-shaped $n$-gon. A drawing $\Gamma$ of $G$ that \emph{extends} $P$ is such that the $n$ vertices of $G$ are placed at the corners of $P$, and every edge of $G$ is drawn either as a side of $P$ or inside $P$.

\begin{lemma}\label{le:extend}
	Let $G$ be an outer simple topological graph with $n \ge 3$ vertices and let $P$ be a star-shaped $n$-gon. There exists a polyline drawing of $G$ with curve complexity at most one that partially preserves the topology of $G$ and that extends $P$.
\end{lemma}
\begin{proof}
	We explain how to compute a drawing with the desired properties for the complete graph $K_n$. Clearly a drawing of $G$ can be obtained by removing the missing edges. Identify each vertex of $K_n$ with a distinct corner of $P$, and let $\{v_0,v_1,\dots,v_{n-1}\}$ be the $n$ vertices of $K_n$ in the clockwise circular order they appear along the boundary of $P$. Note that every edge $(v_i,v_{i+1})$, for $i=0,1,\dots,n-1$ (indices taken modulo $n$), coincides with a side of $P$ and hence it is drawn as a straight-line segment. We now show how to draw all the edges between vertices at \emph{distance} greater than one. The distance between two vertices $v_i$ and $v_j$ is the number of vertices encountered along $P$ when walking clockwise from $v_i$ (excluded) to $v_j$ (included).	We orient each edge $(v_i,v_j)$ from $v_i$ to $v_j$ if the distance between $v_i$ and $v_j$ is smaller than or equal to the distance between $v_j$ and $v_i$. The \emph{span} of an oriented edge $(v_i,v_j)$ is equal to the distance between $v_i$ and $v_j$. We add all oriented edges $(v_i,v_j)$ by increasing value of the span. Let $c$ be an interior point of the kernel, for example its centroid. For any pair of vertices $v_i$ and $v_{j}$, let $b_{i,j}$ be the bisector of the angle swept by $r_i=\overline{cv_i}$ when rotated clockwise around $c$ until it overlaps with $r_{j}=\overline{cv_{j}}$. We denote by $\Gamma_k$ the drawing after the addition of the first $k \ge 0$ edges and maintain the following invariant for $\Gamma_k$.

	\begin{itemize}
		
		\item For each oriented edge $(v_i,v_{j})$ not yet in $\Gamma_k$, there is a point $p_{i,j}$ on $b_{i,j}$ such that $(v_i,v_{j})$ can be drawn with a bend at any point of the segment $\sigma_{i,j}=\overline{cp_{i,j}}$ intersecting any edge of $\Gamma_k$ at most once (either at a crossing or at a common endpoint).
		
	\end{itemize}
	
	We will refer to the segment $\sigma_{i,j}$ described in the invariant as the \emph{free segment} of $(v_i,v_{j})$. Since $P$ is star-shaped, the invariant holds for $\Gamma_0$; in particular the free segment of every $(v_i,v_{j})$ is the intersection of $b_{i,j}$ with the kernel.
	
	\begin{figure}[t]
		\centering
		\begin{minipage}[b]{.3\textwidth}
			\centering 
			\includegraphics[page=1, width=\textwidth]{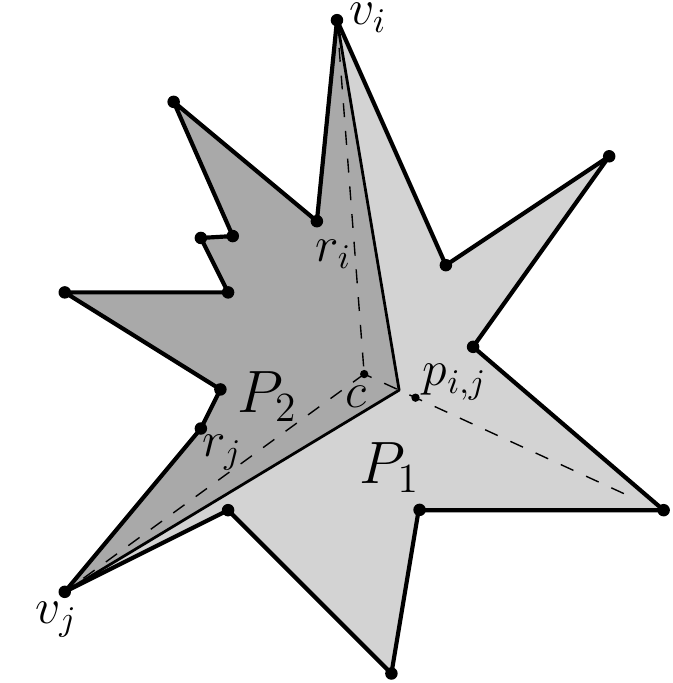}
			\subcaption{~}\label{fi:star-shaped-a}
		\end{minipage}
		\hfil
		\begin{minipage}[b]{.3\textwidth}
			\centering 
			\includegraphics[page=2, width=\textwidth]{figure/star-shaped}
			\subcaption{~}\label{fi:star-shaped-b}
		\end{minipage}
		\hfil
		\begin{minipage}[b]{.3\textwidth}
			\centering 
			\includegraphics[page=3, width=\textwidth]{figure/star-shaped}
			\subcaption{~}\label{fi:star-shaped-c}
		\end{minipage}
		\caption{Illustration for the proof of Lemma~\ref{le:extend}. (a) The two polygons defined by the addition of edge $(v_i,v_j)$. (b) Case 1: $(v_h,v_l)$ is contained in $P_2$. (c) Case 2: $(v_h,v_l)$ intersects $(v_i,v_j)$.\label{fi:star-shaped}}
	\end{figure}
	
	Let $(v_i,v_{j})$ be the $k$-th edge to be added and assume that the invariant holds for $\Gamma_{k-1}$. We place the bend point of $(v_i,v_{j})$ at any point of the segment $\sigma_{i,j}$. By the invariant, the resulting edge intersects any other existing edge at most once. We now prove that the invariant is maintained. The drawing of the edge $(v_i,v_{j})$ divides the polygon $P$ in two sub-polygons (see Fig.~\ref{fi:star-shaped-a}). We denote by $P_1$ the one that contains the portion of the boundary of $P$ that is traversed when going clockwise from $v_i$ to $v_{j}$, and by $P_2$ the other one. Notice that the point $c$ is contained in $P_2$. Let $(v_h,v_{l})$ be any oriented edge not in $\Gamma_k$. Before the addition of $(v_i,v_{j})$, by the invariant there was a free segment $\sigma_{h,l}$ for $(v_h,v_l)$. By construction,  $(v_i,v_{j})$ intersects $\sigma_{h,l}$ at most once. If $(v_i,v_{j})$ and $\sigma_{h,l}$ intersect in a point $p$, let $p'$ be any point between $c$ and $p$ on $\sigma_{h,l}$ and let $\sigma'_{h,l}=\overline{cp'}$; if they do not intersect let $\sigma'_{h,l}=\sigma_{h,l}$. In both cases $\sigma'_{h,l}$ is completely contained in $P_2$. We claim that $\sigma'_{h,l}$ is a free segment for $(v_h,v_l)$. Because of the order used to add the edges, the span of $(v_h,v_l)$ is at least the span of $(v_i,v_j)$. This implies that $v_h$ and $v_{l}$ cannot both belong to $P_1$ (as otherwise the span of $(v_h,v_l)$ would be smaller than the span of $(v_i,v_j)$). We distinguish two cases.
	
	\begin{description}
		\item[Case 1:] Both $v_h$ and $v_{l}$ belong to $P_2$ (possibly coinciding with $v_i$ or $v_{j}$). Refer to Fig.~\ref{fi:star-shaped-b}. For any point $b$ of $\sigma'_{h,l}$, the polyline $\pi$ consisting of the two segments $\overline{v_hb}$ and $\overline{bv_{l}}$ is completely contained in $P_2$ and therefore does not intersects the edge $(v_i,v_{j})$  (except possibly at a common end-vertex if $v_h$ or $v_{l}$ coincide with $v_i$ or $v_{j}$). By the invariant, $\pi$ intersects any other existing edge at most once. Thus, $\sigma'_{h,l}$ is a free segment.
		
		\item[Case 2:] One between $v_h$ and $v_{l}$ belongs to $P_1$ and the other one belongs to $P_2$. Refer to Fig.~\ref{fi:star-shaped-c}. For any point $b$ of $\sigma'_{h,l}$, the polyline $\pi$ consisting of the two segments $\overline{v_hb}$ and $\overline{bv_{l}}$ intersects the edge $(v_i,v_{j})$ exactly once. By the invariant, $\pi$ intersects any other existing edge at most once. Thus, $\sigma'_{h,l}$ is a free segment.
	\end{description}
	
	From the argument above we obtain that the final drawing of $K_n$ has curve complexity one and extends $P$.
	By removing the edges of $K_n$ not in $G$, we obtain a polyline drawing $\Gamma$ of $G$ with curve complexity one that extends $P$.
	Moreover, $\Gamma$ partially preserves the topology of $G$. Namely, the circular order of the edges around each vertex and the external boundary are preserved by construction. Furthermore, since $G$ is outer, any two of its edges cross if and only if their four end-vertices appear interleaved when walking along its external boundary. This property is preserved in $\Gamma$, because the order of the vertices along $P$ is the same as the order of the vertices along the external boundary of $G$, and because any two edges cross at most once (either at a crossing or at a common endpoint).
\end{proof}

We now show how to exploit Lemma~\ref{le:extend} to compute a polyline drawing $\Gamma$ with constant curve complexity for any simple topological graph $G$ that has a biconnected planar skeleton $\sigma(G)$.  


\begin{theorem}\label{th:biconnected-skeleton}
Let $G$ be a simple topological graph that admits a planar augmentation whose planar skeleton is biconnected. Then $G$ has a polyline drawing with curve complexity at most one that partially preserves its topology. The curve complexity is worst-case optimal.
\end{theorem}
\begin{proof}
	Let $G'$ be a planar augmentation of $G$ whose planar skeleton $\sigma(G')$ is biconnected. Each edge of $G' \setminus \sigma(G')$ is inside one face of $\sigma(G')$. Thus, our approach is to compute a drawing of $\sigma(G')$ where each face is drawn as a star-shaped polygon and then to add the missing edges inside each face by using Lemma~\ref{le:extend}. Since the technique of Lemma~\ref{le:extend} has to be slightly adapted to be applied to the external face, we assume first than no edge of $G' \setminus \sigma(G')$ is embedded inside the external face of $\sigma(G')$. We augment $\sigma(G')$ to a suitable planar triangulation by adding a vertex inside each non-triangular internal face $f$ and by connecting it to all the vertices in the boundary of $f$ in a planar way. Computing a straight-line drawing of the augmented graph and removing the dummy vertices and edges, we obtain the desired drawing $\Gamma_{\sigma}'$ whose internal faces are drawn as a start-shaped polygons. Let $f$ be a face of $\sigma(G')$, let $G_f$ be the subgraph of $G'$ consisting of the edges of $f$ plus the edges that are inside $f$, and let $P_f$ be the star-shaped polygon representing $f$ in $\Gamma_{\sigma}'$. By Lemma~\ref{le:extend}, $G_f$ admits a polyline drawing with curve complexity $1$ that weakly preserves the topology of $G_f$ and that extends $P_f$. By computing such a drawing for all faces of $\Gamma_{\sigma}'$ we obtain a polyline drawing $\Gamma'$ of $G'$ with curve complexity $1$. We now prove that $\Gamma'$ partially preserves the topology of $G'$. The drawing $\Gamma_{\sigma}'$ fully preserves the topology of $\sigma(G')$. Since all edges not in $\sigma(G')$ are added inside the face of $\sigma(G')$ in which they are embedded in $G'$, the only case in which the rotation system or the set of crossings could not be preserved is for edges that do not belong to $\sigma(G')$ and that are embedded inside the same face $f$ of $\sigma(G')$. By Lemma~\ref{le:extend} however the drawing of the graph $G_f$ consisting of the edges of each face $f$ plus the edges inside $f$ partially preserves the topology of $G_f$. Thus $\Gamma'$ partially preserves the topology of $G'$ and removing the edges of $G' \setminus G$ we obtain a drawing of $G$ that weakly preserves the topology of $G$.
	
	
	\begin{figure}[tbp]
		\centering
		\includegraphics[width=.2\textwidth]{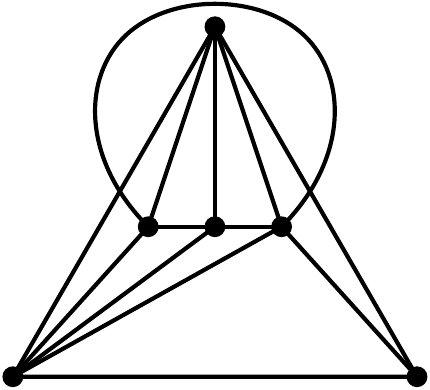}
		\caption{\small
			A simple topological graph with a biconnected planar skeleton that does not admit a straight-line drawing that partially preserves its topology.
		}
		\label{fi:lowerBoundPartially}
	\end{figure}
	
	The argument above assumes that no edge of $G' \setminus \sigma(G')$ is embedded inside the external face of $\sigma(G')$. If there are some edges embedded inside the external face $f^*$ of $\sigma(G')$, we proceed as follows. When triangulating $\sigma(G')$ to obtain a drawing with star-shaped faces, we also triangulate $f^*$ (since there are crossing edges embedded inside $f^*$, it must have degree larger than three). When dummy vertices ad edges are removed from the straight-line drawing of the augmented triangulated graph, we remove all dummy vertices ad edges except those that belong to the external boundary of the augmented graph (they are one vertex and two edges). In this way the resulting drawing has one dummy internal face $f^d$ whose boundary contains all the vertices of $f^*$ (and the only dummy vertex not removed). Face $f^d$ is also star-shaped and thus we can draw inside $f^d$ all the edges $E^*$ that are embedded inside $f^*$ in $G'$ using Lemma~\ref{le:extend}. Removing the dummy vertex and the dummy edges we obtain a drawing where the edges of $E^*$ are drawn inside $f^*$.

	Finally, we show that curve complexity one is optimal in the worst case. The graph of Fig.~\ref{fi:lowerBoundPartially} has a triconnected planar skeleton, and it is immediate to see that it does not admit a straight-line drawing that partially preserves its topology.
\end{proof}

If $\sigma(G)$ is connected, we can draw $G$ with three bends per edge.

\begin{theorem}\label{th:connected-skeleton}
Let $G$ be a simple topological graph that admits a planar augmentation whose planar skeleton is connected. Then $G$ has a polyline drawing with curve complexity at most three that partially preserves its topology.
\end{theorem}

\begin{proof}
	Let $G'$ be a planar augmentation of $G$ whose planar skeleton $\sigma(G')$ is connected. The idea is to add a set $E^*$ of edges to make $\sigma(G')$ biconnected and then use Theorem~\ref{th:biconnected-skeleton}. For each face $f$ (possibly including the external one) whose boundary contains at least one cutvertex we execute the following procedure.  Walk clockwise along the boundary of $f$ and let $v_0, v_1, v_2, \dots , v_k$ be the sequence of vertices in the order they are encountered during this walk, where the vertices that are encountered more than once (i.e. the cutvertices) appear in the sequence only when they are encountered for the first time. For each pair of consecutive vertices $v_{i-1}$ and $v_{i}$ (for $i=1,2,\dots,k$) in the above sequence, if $v_{i-1}$ and $v_i$ are not adjacent in $\sigma(G')$, add to $E^*$ the edge $(v_{i-1},v_i)$. See Fig.~\ref{fi:connected-a} and~\ref{fi:connected-c} for an example.
	
	
	With the addition of the edges of $E^*$, $\sigma(G')$ becomes biconnected (the boundary of each face is a simple cycle). In particular every added edge $(v_{i-1},v_i)$ connects vertices of two different biconnected components, and for every pair of biconnected components there is at most one edge of $E^*$ that connects them.

	\begin{figure}[tbp]
		\centering
		\begin{minipage}[b]{.28\textwidth}
			\centering 
			\includegraphics[page=1, width=\columnwidth]{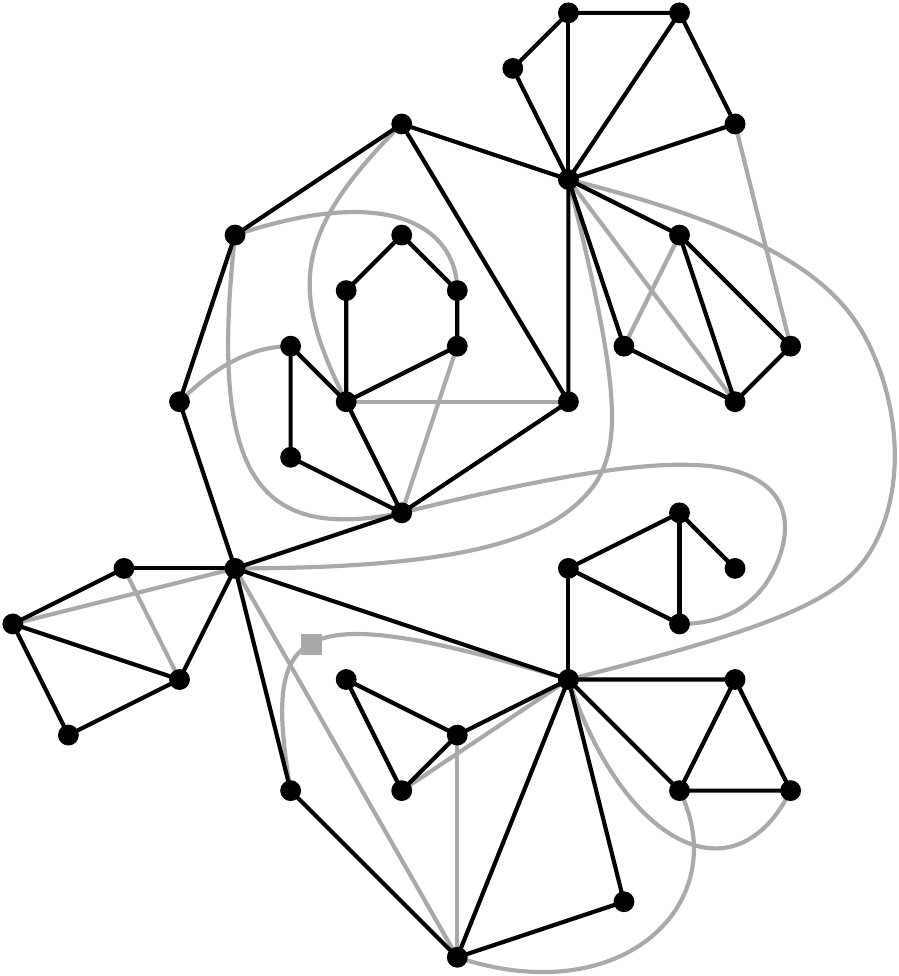}
			\subcaption{~}\label{fi:connected-a}
		\end{minipage}
		\hfil
		\begin{minipage}[b]{.28\textwidth}
			\centering 
			\includegraphics[page=3, width=\textwidth]{figure/connected}
			\subcaption{~}\label{fi:connected-c}
		\end{minipage}
		\hfil
		\begin{minipage}[b]{.28\textwidth}
			\centering 
			\includegraphics[page=4, width=\textwidth]{figure/connected}
			\subcaption{~}\label{fi:connected-d}
		\end{minipage}
		\caption{(a) A simple topological graph $G$. The planar skeleton $\sigma(G)$ of $G$ is shown in black. (b)  Augmentation of $\sigma(G)$ to make it biconnected. (d) Augmentation of $G$. Each edge of $G \setminus \sigma(G)$ (in gray) is crossed by the augmenting edges at most twice.}
	\end{figure}

	If we add the edges of $E^*$  to $G'$ (embedded in the same way with respect to $\sigma(G')$), we obtain a new topological graph such that the edges of $E^*$ cross the edges of $G' \setminus \sigma(G)$ (see Fig.~\ref{fi:connected-d}). In particular, the edges of $G' \setminus \sigma(G)$ that are crossed by the edges of $E^*$ are those incident to the cutvertices of $\sigma(G')$. Let $e=(u,v)$ be one such edge and suppose that $u$ is a cutvertex. In the circular order of the edges around $u$, the edge $e$ appears between two different biconnected components of $\sigma(G')$ sharing $u$; if $E^*$ contains an edge $e'$ connecting these two components, then $e'$ crosses $e$. Notice that $e'$ can be embedded in such a way that the crossing $c$ between $e'$ and $e$ is the first one encountered along $e$ when going from $u$ to $v$. In other words, the portion of $e$ from $u$ to $c$ is not crossed. Since both end-vertices of an edge can be cutvertices, each edge of $G' \setminus \sigma(G)$ is crossed by the edges of $E^*$ at most twice. Replacing each of the crossings created by the addition of $E^*$ with dummy vertices, we obtain a new topological graph $G''$ whose planar skeleton is biconnected. By Theorem~\ref{th:biconnected-skeleton} $G''$ admits a drawing that partially preserves its topology and such that each edge has at most one bend. Replacing dummy vertices with bends, we obtain a drawing of $G'$ that partially preserves its topology. We now show that the number of bends per edge is at most $3$. Let $e$ be any edge of $G'$. As described above, $e$ is crossed at most twice by the edges of $E^*$ and therefore $e$ is split in at most three ``pieces'' in $G''$. The two ``pieces'' that are incident to the original vertices are not crossed in $G''$ and therefore they belong to $\sigma(G'')$ and are drawn without bends. The third ``piece'' is not in $\sigma(G'')$ and is drawn with at most one bend. Thus, $e$ has at most three bends.
\end{proof}

 Theorems~\ref{th:biconnected-skeleton} and~\ref{th:connected-skeleton} show that constant curve complexity is sufficient for drawings that partially preserve the topology of graphs whose planar skeleton is connected. It is worth remarking that a drawing that fully preserves the topology may require $\Omega(n)$ curve complexity even if the planar skeleton is connected.
 Namely, the planar skeleton of the graphs associated with arrangements of pseudolines is always biconnected and, by Corollary~\ref{co:eppstein}, there exists one such graph that has $\Omega(n)$ curve complexity.

%

 One may wonder whether the constant curve complexity bound of Theorems~\ref{th:biconnected-skeleton} and~\ref{th:connected-skeleton} can be extended to the case of non-connected planar skeletons. This question is answered in the negative by the next theorem.

\begin{theorem}\label{th:disconnected-lower-bound}
	There exists a simple topological graph with $n$ vertices such that any drawing that partially preserves its topology has curve complexity $\Omega(\sqrt{n})$.
\end{theorem}
\begin{proof}
	Let $\mathcal{L}$ be an arrangement of pseudolines and let $G_{\mathcal{L}}$ be the graph associated with $\mathcal{L}$. By Lemma~\ref{le:eppstein} any drawing that fully preserves the topology of $G_{\mathcal{L}}$ cannot have a better curve complexity than $\mathcal{L}$. On the other hand if we only want to partially preserve the topology, $G_{\mathcal{L}}$ can be realized without bends (see Fig.~\ref{fi:lowerBoundDisconnected-a} for a straight-line drawing of the graph of Fig.~\ref{fi:arrangement-b}). We now describe how to construct a supergraph $\overline{G}_{\mathcal{L}}$ of $G_{\mathcal{L}}$, such that in any drawing of $\overline{G}_{\mathcal{L}}$ that partially preserves its topology, the topology of the subgraph $G_{\mathcal{L}}$ is fully preserved. Refer to Fig.~\ref{fi:lowerBoundDisconnected-b} for an illustration concering the graph of Fig.~\ref{fi:arrangement-b}. The set $E^*$ of crossing edges of $G_{\mathcal{L}}$ form a set of cells inside the cycle $C$ of $G_{\mathcal{L}}$ (these cells correspond to the faces of the planarization of $G_{\mathcal{L}}$ that have at least one dummy vertex). For each of these cells, we add a vertex inside the cell and we connect two such vertices if the corresponding cells share a side. For those cells that have as a side an edge $e$ of $C$ we add an edge between the vertex added inside that cell and the two end-vertices of $e$. Let $\overline{G}_{\mathcal{L}}$ be the resulting topological graph and let $\overline{\Gamma}_{\mathcal{L}}$ be a drawing that partially preserves the topology of $\overline{G}_{\mathcal{L}}$. We claim that the sub-drawing $\Gamma_{\mathcal{L}}$ of $\overline{\Gamma}_{\mathcal{L}}$ representing $G_{\mathcal{L}}$ fully preserves the topology of $G_{\mathcal{L}}$. If we remove the edges in $E^*$, we obtain a planar subgraph $G'$ whose sub-drawing $\Gamma'$ in $\overline{\Gamma}_{\mathcal{L}}$ is planar. By construction, any two faces of $G'$ share at most one edge or at most one vertex. By Barnette's Theorem~\cite{barnette94} $G'$ is triconnected and therefore it has only one planar embedding, which is the one defined by $\overline{G}_{\mathcal{L}}$. Let $e=(u,v)$ be an edge of $E^*$. In $G'$ (and therefore in $\Gamma'$) there exists two paths $\pi_1=\langle u, u_1, u_2, \dots, u_k, v \rangle$ and $\pi_2=\langle u, v_1, v_2, \dots, v_k, v \rangle$ from $u$ to $v$ with the edges $e_i=(u_i,v_i)$ for every $i=1,2,\dots,k$. In $\overline{G}_{\mathcal{L}}$ the edge $e$ crosses the edges $e_1, e_2, \dots, e_k$ in this order and it crosses no other edge. Since the set of crossings is preserved in $\overline{\Gamma}_{\mathcal{L}}$ and there is no way for $e$ to cross the edges $e_1, e_2, \dots, e_k$ in a different order without creating another crossing, then $e$ cross $e_1, e_2, \dots, e_k$ in $\overline{\Gamma}_{\mathcal{L}}$ in the same order as in $\overline{G}_{\mathcal{L}}$. This is true for every edge of $E^*$, which implies that also the crossings between the edges of $E^*$ are preserved in the same order and that $\Gamma_{\mathcal{L}}$ is a drawing that fully preserves the topology of $G_{\mathcal{L}}$.
	
\begin{figure}[t]
	\centering
	\begin{minipage}[b]{.28\textwidth}
		\centering 
		\includegraphics[page=1, width=\textwidth]{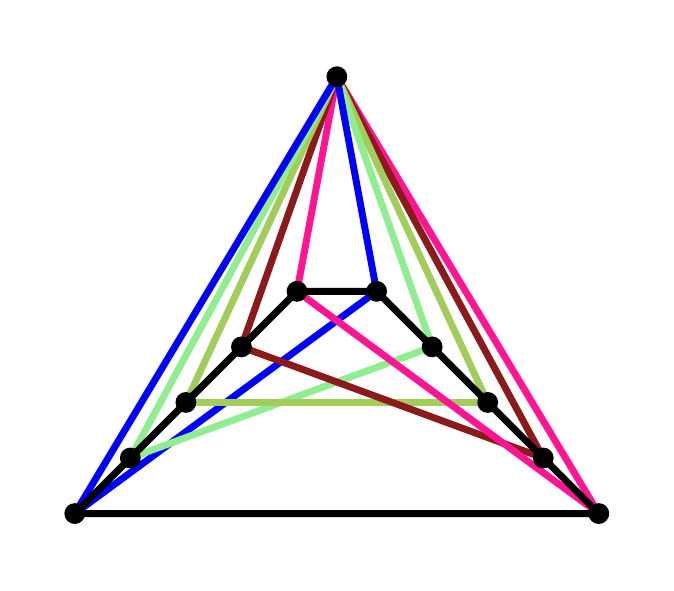}
		\subcaption{~}\label{fi:lowerBoundDisconnected-a}
	\end{minipage}
	\hfil
	\begin{minipage}[b]{.28\textwidth}
		\centering 
		\includegraphics[page=2, width=\textwidth]{figure/lowerBoundDisconnected}
		\subcaption{~}\label{fi:lowerBoundDisconnected-b}
	\end{minipage}
	\caption{(a) Straight-line drawing of the graph $G_{\mathcal{L}}$ of Fig.~\ref{fi:arrangement-b}. (b) The graph $\overline{G}_{\mathcal{L}}$ for the arrangement of Fig.~\ref{fi:arrangement-a}.\label{fi:lowerBoundDisconnected}}
\end{figure}
	
	Denote by $\mathcal{L}_N$ the arrangement of $N$ pseudolines defined by Eppstein et al.~\cite{DBLP:journals/corr/EppsteinGSU16}. By the argument above, any polyline drawing that partially preserves the topology of the graph $\overline{G}_{\mathcal{L}_N}$ contains a sub-drawing of $G_{\mathcal{L}_N}$ that fully preserves its topology and that therefore has curve complexity $\Omega(N)$ by Lemma~\ref{le:eppstein}. The number of vertices of $G_{\mathcal{L}_N}$ is $2N+1$ and the number of cells is $\Theta(N^2)$. This implies that the number of vertices of $\overline{G}_{\mathcal{L}_N}$ is $n=\Theta(N^2)$. Thus, any drawing that partially preserves the topology of $\overline{G}_{\mathcal{L}_N}$ has curve complexity $\Omega(N)=\Omega(\sqrt{n})$.
\end{proof}

Based on Theorem~\ref{th:disconnected-lower-bound} one may wonder whether $O(\sqrt{n})$ curve complexity is sufficient when the skeleton is not connected. The following theorem states a preliminary result in this direction, extending Theorem~\ref{th:connected-skeleton} to the case that the planar skeleton consists of at most $c$ connected components.

\begin{theorem}\label{th:non-connected-skeleton}
	Let $G$ be a simple topological graph that admits a planar augmentation whose planar skeleton has $c$ connected components. Then $G$ has a polyline drawing with curve complexity at most $4c-1$ that partially preserves its topology.
\end{theorem}
\begin{proof}
	We can assume that $G$ is connected. If not we can compute a drawing for each connected component. Let $G'$ be a planar augmentation of $G$ whose planar skeleton $\sigma(G')$ has $c$ connected components. Since $G$ is connected, there exists a set of edges of $G \setminus \sigma(G')$ that can be added to $\sigma(G')$ to make it connected. In particular, we can choose a set $E'$ with $c-1$ of these edges. Denote by $G''$ the graph obtained by adding the edges of $E'$ to $G'$. The edges of $E'$ can cross each other. If this is the case we replace each crossing between two edges of $E'$ with a dummy vertex, thus obtaining a new graph $G'''$. Denote by $E''$ the set of edges obtained by the subdivision of the edges in $E'$. Since each edge of $E'$ is crossed at most $c-2$ times, the set $E''$ has at most $(c-1)^2$ edges. We now use the sleeve method (see Section~\ref{se:fully}): we put a sleeve around each edge of $E''$. Let $G^{iv}$ be the resulting graph. The planar skeleton $\sigma(G^{iv})$ of $G^{iv}$ is connected since it contains all the edges of the original skeleton $\sigma(G')$ and all edges of the sleeves, which connected the different connected components of $\sigma(G')$. By Theorem~\ref{th:connected-skeleton} $G^{iv}$ admits a polyline drawing with curve complexity three that partially preserves its topology. Replacing the dummy vertices with bends and removing the dummy edges of the sleeves we obtain a drawing of $G'$ that partially preserves its topology. We claim that the curve complexity of this drawing is at most $4c-1$. Let $e$ be an edge of $G'$. If $e$ is an edge of $\sigma(G')$ is drawn without bends. If $e \in E'$, then $e$ is split in $G^{iv}$ in at most $c-1$ ``pieces''. Each ``piece'' has at most one bend and at most $c-2$ additional bends are created by the dummy vertices that split $e$, thus the total number of bends is $2c-3$. If $e$ does not belong to $\sigma(G')$ nor to $E'$, let $k$ be the number of sleeves traversed by $e$. Then $e$ is subdivided in $2k+1$ ``pieces'' (one for each sleeve and $k+1$ outside the sleeves). There is one bend for each ``piece'' plus one for each dummy vertex splitting $e$, thus the number of bends is $2k+1+2k=4k+1$. According to the technique of Theorem~\ref{th:connected-skeleton}, edge $e$ can be split by two additional dummy vertices at each end if it is incident to two cut-vertices of the planar skeleton, thus the number of bends is at most $4k+3$. The number of sleeves is at most $(c-1)^2$ (i.e., the size of $E''$), but the edge $e$ can traverse at most $c$ of them, because it crosses each edge of $E'$ at most once; thus $k \leq c-1$ and therefore the curve complexity is at most $4(c-1)+3=4c-1$.\end{proof}

\section{Polyline Drawings that Fully Preserve the Topology}\label{se:fully}

In this section we study polyline drawings of constant curve complexity for two meaningful families of beyond-planar graphs. Namely, we consider \emph{$k$-skew} graphs and \emph{$2$-plane} graphs. A simple topological graph $G = (V,E)$ is  $k$-skew if there is a set $F \subseteq E$  of $k$ edges such that $G' = (V, E \setminus F)$ does not contain crossings.  A simple topological graph is $2$-plane if every edge is crossed by at most two other edges. A $2$-plane graph with $n$ vertices can have at most $5n - 10$ edges and it is called \emph{optimal $2$-plane} if it has exactly $5n - 10$ edges.
We  prove that the graphs belonging to these two families admit a polyline drawing that fully preserves the topology and has constant curve complexity. A tool that we are going to use is the algorithm of Chiba et al.~\cite{Chiba1985} that receives as input a $3$-connected plane graph $G$ whose external face has $k \ge 3$ vertices, and  a convex polygon $P$ with $k$ corners. The algorithm computes a straight-line drawing $\Gamma$ of $G$ that fully preserves the topology of $G$, it has  polygon $P$ as its external face, and all internal faces are convex. Moreover, if three consecutive vertices belong to a same face and are collinear in the computed drawing, we can slightly perturb one of them without destroying the convexity of the other faces. Thus, we can assume that all faces of $\Gamma$ are strictly convex.

\subsection{$k$-skew Topological Graphs}

We first show that a $k$-skew topological graph admits a polyline drawing that fully preserves the topology of $G$ and has at most $2k$ bends per edge. The technique is based on an approach that we call the \emph{sleeve method} and that is illustrated in the following.

\smallskip
\noindent \textbf{The sleeve method.} Suppose that $G$ is a topological graph such that the removal of the edge $(s,t)$ makes $G$ without crossings such as in Fig.~\ref{fi:sleeve-a}.

\begin{figure}[h!]
	
	\centering
	\begin{minipage}[b]{.24\textwidth}
		\centering 
		\includegraphics[page=1, width=\textwidth]{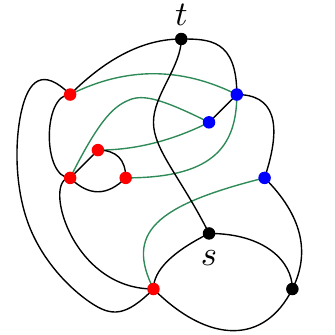}
		\subcaption{}\label{fi:sleeve-a}
	\end{minipage}\hfil
	\begin{minipage}[b]{.24\textwidth}
		\centering 
		\includegraphics[page=2, width=\textwidth]{figure/sleeveProcess2-2}
		\subcaption{}\label{fi:sleeve-b}
	\end{minipage}\hfil
	\begin{minipage}[b]{.24\textwidth}
		\centering 
		\includegraphics[page=3, width=\textwidth]{figure/sleeveProcess2-2}
		\subcaption{}\label{fi:sleeve-c}
	\end{minipage}\hfil
	\begin{minipage}[b]{.24\textwidth}
		\centering 
		\includegraphics[page=4, width=\textwidth]{figure/sleeveProcess2-2}
		\subcaption{}\label{fi:sleeve-d}
	\end{minipage}
	\caption{
		(a) A 1-skew topological graph $G$;  deletion of the edge $(s,t)$ gives a planar topological graph. (b) Two dummy vertices are added to each edge in $E_\chi$. The dummy vertices become left and right vertices; the previously left and right vertices are now neither left nor right.
		(c) Paths $p_L$ (colored red) and $p_R$ (colored blue) are added. (d) The graph obtained from $G$ by adding the ``sleeve''.
	}\label{fi:sleeve}
\end{figure}

Let $E_\chi$ be the set of edges that cross $(s,t)$ and suppose that $\alpha$ is a crossing between edges $(s,t)$ and $(u,v) \in E_\chi$ in $G$. If the clockwise order of the vertices around $\alpha$ is $\langle
s, u , t , v \rangle$, then $u$ is a \emph{left} vertex and $v$ is a \emph{right} vertex (with respect to the ordered pair $(s,t)$ and the crossing $\alpha$). This is illustrated in Fig.~\ref{fi:sleeve-a}: left vertices are coloured red, and right vertices are coloured blue.
We add a ``sleeve'' around $(s,t)$, as follows (refer to Fig.~\ref{fi:sleeve-b}).
Number the edges of $E_\chi= \{ e_1, e_2, \ldots , e_p \}$ in the order of their crossings $\alpha_1 ,  \alpha_2 , \ldots , \alpha_p$ along $(s,t)$, so that $e_i = (u_i, v_i)$ crosses $(s,t)$ at $\alpha_i$, $u_i$ is left, and $v_i$ is right.
We subdivide each edge $(u_i, v_i)$
with dummy vertices $u'_i$ and $v'_i$ so that the edge $(u_i,v_i)$ becomes a path $(u_i, u'_i, v'_i, v_i)$ with the crossing point $\alpha_i$  in between $u'_i$ and $v'_i$. Note that after this subdivision, $u'_i$ is left and $v'_i$ is right, and $u_i$ and $v_i$ are neither left nor right. Next we add a path $p_L$ that begins at $s$ and visits each of the left dummy vertices $u'_i$ in the order $u_1 , u_2 , \ldots , u_p$, and ends at $t$. Similarly we add a path $p_R$ that visits $s$, all the right vertices, and then $t$.
This is illustrated in Fig.~\ref{fi:sleeve-c}.
We call the cycle formed by $p_L$ and $p_R$ a \emph{sleeve}.
Note that the interior of the sleeve contains the edges $(u'_i, v'_i)$ and the edge $(s,t)$, but no other vertices or edges (Fig.~\ref{fi:sleeve-d}).
%
%
%
The next theorem explains how to draw $k$-skew graphs with curve complexity $2k$.

\begin{theorem}\label{th:k-skew}
Every $k$-skew simple topological graph admits a polyline drawing with curve complexity at most $2k$ that fully preserves its topology.
\end{theorem}
\begin{proof}
Suppose that  $G = (V,E)$ is a topological graph and there is a set $F \subseteq E$ of $k$ edges such that deleting all the edges in $F$ from $G$ gives a planar topological graph. An example with $k=2$ is in Fig.~\ref{fi:kskewA-a}.
\begin{figure}[tbp]
  \centering
  \begin{minipage}[b]{.18\textwidth}
		\centering 
		\includegraphics[page=1, width=\textwidth]{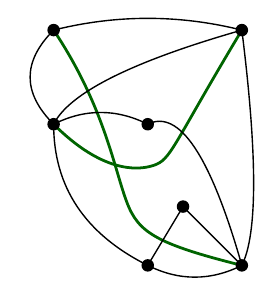}
		\subcaption{~}\label{fi:kskewA-a}
	\end{minipage}\hfil
	\begin{minipage}[b]{.18\textwidth}
		\centering 
		\includegraphics[page=2, width=\textwidth]{figure/kskewA-2}
		\subcaption{~}\label{fi:kskewA-b}
	\end{minipage}\hfil
	\begin{minipage}[b]{.18\textwidth}
		\centering 
		\includegraphics[page=3, width=\textwidth]{figure/kskewA-2}
		\subcaption{~}\label{fi:kskewA-c}
	\end{minipage}
  \caption{
(a) A topological graph $G$ with a set $F$ of 2 edges (in green) whose deletion makes $G$ planar.
(b) A topological graph $G''$ formed from $G$ by splitting the edges of $F$ with a dummy vertex and adding a sleeve around each portion of the split edges.
(c) The graph obtained by deleting the interior of each sleeve in $G''$ and triangulating the graph except for the faces formed by the sleeves.
}\label{fi:kskewA}
\end{figure}
Replace each crossing between a pair of edges in $F$ with a dummy vertex, and let $G'$ be the resulting graph. In $G'$ there is a set $F'$ of edges such that no two edges in $F'$ cross, and deleting all the edges in $F'$ from $G'$ gives a planar topological graph. Here $|F'| \leq k + 2c$, where $c$ is the number of crossings between edges in $F$. Also, note that the number of such crossings on each edge in $F$ is at most $k-1$. Now add a sleeve around each edge $(s,t) \in F'$ using the sleeve method, and let $G''$ be the resulting graph (see Fig.~\ref{fi:kskewA-b}).
Note that two such sleeves do not share any edge, and they share at most one vertex. Delete the interior of each sleeve in $G''$ to give a planar topological graph $G'''$.
Note that each sleeve of $G''$ gives a face of $G'''$. Now triangulate $G'''$ except for the faces of $G'''$ formed by the sleeves (see Fig.~\ref{fi:kskewA-c}).

The resulting graph $G^{iv}$ is triconnected by Barnette's Theorem~\cite{barnette94}, since two faces share at most one edge or at most one vertex. We can construct a planar drawing $\Gamma^{iv}$ of $G^{iv}$ using the convex drawing algorithm of Chiba et al.~\cite{Chiba1985}. Each face of $\Gamma^{iv}$ is convex, including each face that comes from a sleeve. Drawing the edges of $G''$ inside each sleeve as straight-line segments gives a straight-line drawing of $G''$.

\begin{figure}[h!]
	\centering
	\begin{minipage}[b]{.3\textwidth}
		\centering 
		\includegraphics[page=1, width=\textwidth]{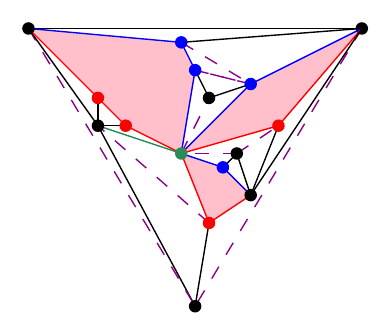}
		\subcaption{}\label{fi:kskewB-a}
	\end{minipage}\hfil
	\begin{minipage}[b]{.3\textwidth}
		\centering 
		\includegraphics[page=2, width=\textwidth]{figure/kskewB-2}
		\subcaption{}\label{fi:kskewB-b}
	\end{minipage}\hfil
	\begin{minipage}[b]{.3\textwidth}
		\centering 
		\includegraphics[page=3, width=\textwidth]{figure/kskewB-2}
		\subcaption{}\label{fi:kskewB-c}
	\end{minipage}
	\captionof{figure}{
		(a) A convex drawing $\Gamma^{iv}$ of $G^{iv}$, from the algorithm of Chiba et al.~\cite{Chiba1985}. (b) Drawing the edges of $G''$ inside each face obtained from a sleeve.
		(c) Removing the dummy edges and vertices gives a drawing of $G$ with at most $2k$ bends on each edge.
	}
	\label{fi:kskewB}
\end{figure}

Deleting the dummy edges of the sleeves, and replacing the dummy vertices of the sleeves by bends, we have a polyline drawing $\Gamma$ of $G$ that fully preserves the embedding of $G$.
The only bends are (1) at the crossing points between edges of $F$, and (2) at the dummy vertices of the sleeves. Let $e$ be an  edge of $G$. If $e \in E\setminus F$, then $e$ crosses at most $k$ edges (those in $F$) and each of these crossings creates two dummy vertices in a sleeve of $G''$, thus resulting in $2k$ bends. If $e \in F$, then it has bends at the crossings with other edges of $F$, which are at most $k-1$. An example of this procedure is depicted in Fig.~\ref{fi:kskewB}.
\end{proof}

By Theorem~\ref{th:k-skew} we can draw a $1$-skew topological graph with two bends per edge. We now  prove that these graphs can be drawn using only one bend per edge. To this aim we first recall some results from~\cite{eades2015}. We say that a vertex is \emph{inconsistent} with respect to the edge $(s,t)$ if it is both left and right with respect to $(s,t)$, and \emph{consistent} otherwise. 
For example, the graph in Fig.~\ref{fi:inconsistent-a} has an inconsistent vertex.
\begin{figure}[t]
  \centering
  \begin{minipage}[b]{.18\textwidth}
		\centering 
		\includegraphics[page=1, width=\textwidth]{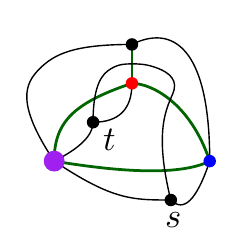}
		\subcaption{~}\label{fi:inconsistent-a}
	\end{minipage}\hfil
	\begin{minipage}[b]{.18\textwidth}
		\centering 
		\includegraphics[page=2, width=\textwidth]{figure/inconsistent-2}
		\subcaption{~}\label{fi:inconsistent-b}
	\end{minipage}
  \caption{(a) A 1-skew graph with an inconsistent vertex (larger and purple).
(b) A 1-skew graph with an internal inconsistent face (shaded), in which every vertex is consistent.}
  \label{fi:inconsistent}
\end{figure}
Observe that in a straight-line drawing of a topological graph,
an inconsistent vertex would have to be both left and right of the straight line through $s$ and $t$.
This gives the following necessary condition.
\begin{lemma}
\label{le:leftRight}
\emph{\cite{eades2015}}
\label{le:inconsistentVertex}
A 1-skew simple topological graph with an inconsistent vertex has no straight-line drawing that fully preserves its topology.
\end{lemma}

Without additional assumptions, the converse of Lemma~\ref{le:leftRight} is false.
For an example,
consider Fig.~\ref{fi:inconsistent-b}; this graph has no straight-line drawing, even though all vertices are consistent.
The problem is that the \emph{internal face} $(s,u,t,v)$ has both left and right vertices; as such,
this \emph{face} is \emph{inconsistent}. To explore the converse of Lemma~\ref{le:leftRight}, we
can assume that the topological graph is \emph{maximal} 1-skew (that is, no edge can be added while retaining the property of being 1-skew). Namely,  it has been proven that every 1-skew simple topological graph $G$ with no inconsistent vertices can be augmented with dummy edges so that the resulting graph has no inconsistent vertices, it is maximal $1$-skew, and it fully preserves the topology of its subgraph $G$~\cite{eades2015} .
%
%
Note that both the simple topological graphs in Fig.~\ref{fi:inconsistent} are maximal 1-skew. We denote the set of left (resp. right) vertices of a 1-skew topological graph $G$ by $V_L$ (resp. $V_R$), the subgraph of ${G}$ induced by $V_L \cup \{s,t \}$ (resp. $V_R \cup \{s,t \}$) by $G_L$ (resp. ${G}_R$), the union of ${G}_L$ and ${G}_R$ by $G_{LR}$. Note that $G_L$ and $G_R$ are induced subgraphs, but $G_{LR}$ is \emph{not necessarily induced} as a subgraph of $G$.
%
%
The following is proved in~\cite{eades2015}.

\begin{lemma}
\label{th:1skewMax}
\emph{\cite{eades2015}}
Let $G$ be a maximal 1-skew graph with all vertices consistent. Then:
\begin{inparaenum}

\item [(a)] $G_{LR}$ has exactly one inconsistent face, and this face contains both $s$ and $t$; and

\item [(b)] $G$ has a straight-line drawing that fully preserves its topology if and only if the inconsistent face of $G_{LR}$ is the external face (of $G_{LR}$).

\end{inparaenum}
\end{lemma}


Let $(s,t)$ be the edge of $G$ whose removal makes $G$ planar. It is clear that after adding a sleeve around edge $(s,t)$, the conditions of Lemma~\ref{th:1skewMax} are satisfied and thus, we can compute a straight-line drawing, which after removing the dummy vertices of the sleeve, gives rise to a drawing with at most two bends per edge. To prove that one bend per edge suffices, we need a more subtle argument.

\begin{figure}[t]
  \centering
  \begin{minipage}[b]{.18\textwidth}
		\centering 
		\includegraphics[page=1, width=\textwidth]{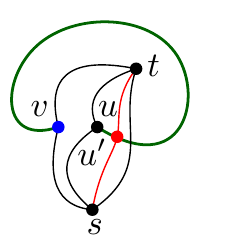}
		\subcaption{~}\label{fi:halfSleeveGLR-a}
	\end{minipage}\hfil
	 \begin{minipage}[b]{.18\textwidth}
		\centering 
		\includegraphics[page=2, width=\textwidth]{figure/halfSleeveGLR-2}
		\subcaption{~}\label{fi:halfSleeveGLR-b}
	\end{minipage}\hfil 
	\begin{minipage}[b]{.18\textwidth}
		\centering 
		\includegraphics[page=3, width=\textwidth]{figure/halfSleeveGLR-2}
		\subcaption{~}\label{fi:halfSleeveGLR-c}
	\end{minipage}\hfil 
	\begin{minipage}[b]{.18\textwidth}
		\centering 
		\includegraphics[page=4, width=\textwidth]{figure/halfSleeveGLR-2}
		\subcaption{~}\label{fi:halfSleeveGLR-d}
	\end{minipage}
  \caption{
(a) A left half-sleeve is added to the graph $G$ in Fig.~\ref{fi:inconsistent}(b) to form $G^{*L}$.
(b) $G^{*L}_{LR}$ has an internal inconsistent face.
(c) A right half-sleeve is added to the graph $G$ in Fig.~\ref{fi:inconsistent}(b) to form $G^{*R}$.
(d) $G^{*R}_{LR}$ has no internal inconsistent face.
  }
  \label{fi:halfSleeveGLR}
\end{figure}

\begin{theorem}\label{th:1skew}
	Every $1$-skew simple topological graph admits a polyline drawing with curve complexity at most one that fully preserves its topology. The curve complexity is worst-case optimal.
\end{theorem}
\begin{proof}
	Instead of placing a sleeve around the edge $(s,t)$, we use a ``half-sleeve'', as follows. Again let $E_\chi$ be the set of edges that cross $(s,t)$. We 1-subdivide each edge $(u,v) \in E_\chi$ with a dummy vertex on the left side of the crossing that $(u,v)$ makes with $(s,t)$, then add a path $p_L$ that begins at $s$ and visits each of the left dummy vertices in the order that there incident edges cross $(s,t)$, and ends at $t$. Denote the graph obtained from $G$ by adding this ``left half-sleeve'' as above by $G^{*L}$. Similarly, we could add a ``right half-sleeve'' to obtain a topological graph $G^{*R}$.
	It is clear that every vertex in both $G^{*L}$ and $G^{*R}$ is consistent. Note also that we have only added one dummy vertex on each edge $(u,v) \in E_\chi$; we aim to draw each of these edges with only one bend per edge. However, it is not clear that the internal faces of $G^{*L}_{LR}$ and $G^{*R}_{LR}$ are consistent.
	Consider, for example, the graph $G$ in Fig.~\ref{fi:inconsistent}(b).
	For this graph, Fig~\ref{fi:halfSleeveGLR} shows  $G^{*L}$, $G^{*R}$,  $G^{*L}_{LR}$ and $G^{*R}_{LR}$.
	Note that $G^{*L}_{LR}$ has an internal inconsistent face, while  $G^{*L}_{LR}$ does not.

	In order to prove  that at most one of the graphs $G^{*L}_{LR}$  and $G^{*R}_{LR}$ has an internal inconsistent face,
	suppose that the inconsistent face $f$ of $G^{*L}_{LR}$ is internal. Note that all the left vertices of $G^{*L}_{LR}$ lie on the path $p_L$,
	and so $p_L$ forms part of the boundary of $f$. Thus $f$ consists of $p_L$, then a walk $w_R$ of right vertices that begins at $t$ and ends at $s$.
	If we traverse $f$ in a clockwise direction, the interior is on the right of each edge.
	Note that the edge $(s,t)$ lies outside this face.
	Further, an edge $(u',v)$ from $E_\chi$ that crosses $(s,t)$ has a left vertex on the left and a right vertex on the right;
	this is illustrated in Fig.~\ref{fi:insideFace1-a}.
	\begin{figure}[t]
		\centering
		\begin{minipage}[b]{.3\textwidth}
			\centering 
			\includegraphics[page=1, width=\textwidth]{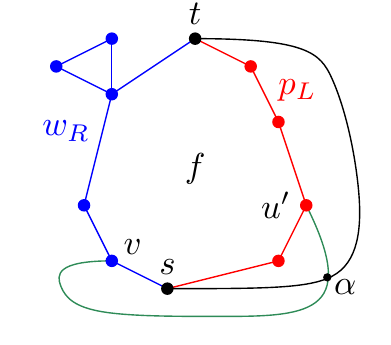}
			\subcaption{}\label{fi:insideFace1-a}
		\end{minipage}\hfil
		\begin{minipage}[b]{.3\textwidth}
			\centering 
			\includegraphics[page=2, width=\textwidth]{figure/insideFace1-2}
			\subcaption{}\label{fi:insideFace1-b}
		\end{minipage}\hfil
		\begin{minipage}[b]{.3\textwidth}
			\centering 
			\includegraphics[page=3, width=\textwidth]{figure/insideFace1-2}
			\subcaption{}\label{fi:insideFace1-c}
		\end{minipage}
		\caption{
			(a) The face $f$ consists of the path $p_L$ from $s$ to $t$, then a walk $w_R$ of right vertices that begins at $t$ and ends at $s$.
			(b) Every vertex $u$ that is a left vertex in $G$ lies inside the cycle $c$ formed by $q_R$ and the edge $(s,t)$.
			(c) The inconsistent face of $G^{*R}_{LR}$ consists of the path $p_R$, plus a path $q_L$ of vertices that are left in $G$.
		}
		\label{fi:insideFace1}
	\end{figure}
	Now the walk $w_R$ may not be a simple path, but it contains a simple path $q_R$; concatenating $p_L$ with $q_R$ gives a simple cycle.
	Note that in $G^{*L}$, every vertex $u$ that is a left vertex in $G$ is inside this cycle in $G^{*L}$. Thus in $G$,
	every vertex $u$ that is a left vertex in $G$ is inside the cycle $c$ formed by $q_R$ and the edge $(s,t)$.
	This is illustrated in Fig.~\ref{fi:insideFace1-b}.
	It follows that every path $q_L$ of left vertices from $s$ to $t$ lies inside the cycle formed by $q_R$ and the edge $(s,t)$. Now consider the graph $G^{*R}$. The inconsistent face of $G^{*R}_{LR}$ consists of the path $p_R$,
	plus a path $q_L$ of vertices that are left in $G$. It is clear that this is the outside face of $G^{*R}_{LR}$; see Fig.~\ref{fi:insideFace1-c}. This concludes the proof.
\end{proof}

\subsection{Optimal $2$-plane Graphs}

Let $G$ be a simple optimal $2$-plane graph with $n$ vertices. Bekos et al.~\cite{DBLP:conf/compgeom/Bekos0R17} proved that the planar skeleton $\sigma(G)$ of $G$ is a \emph{pentangulation} with $n$ vertices, i.e., each face of $\sigma(G)$ is a simple $5$-cycle, which we call \emph{pentagon}, and $\sigma(G)$ spans all the vertices of $G$. See, for example, Fig.~\ref{fi:o2p-g}. Moreover, each face of $\sigma(G)$ has five crossing edges in its interior, which we call \emph{chords} in the following. Bekos et al. proved that $\sigma(G)$ is always $2$-connected; we can prove that it is actually $3$-connected.

\begin{figure}[t]
	\centering
	\begin{minipage}[b]{.3\textwidth}
		\centering 
		\includegraphics[page=1, width=\textwidth]{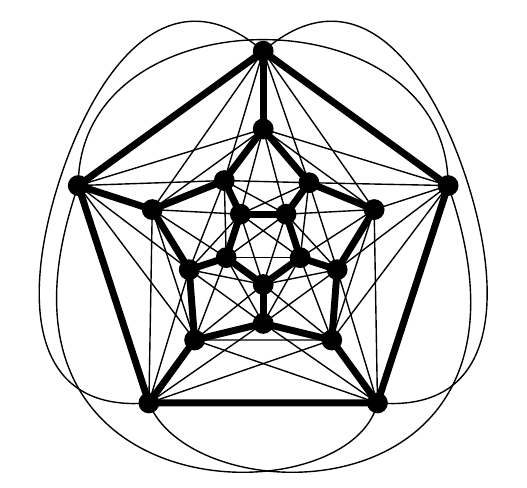}
		\subcaption{~}\label{fi:o2p-g}
	\end{minipage}
	\hfil
	\begin{minipage}[b]{.3\textwidth}
		\centering 
		\includegraphics[page=2, width=\textwidth]{figure/o2p}
		\subcaption{~}\label{fi:o2p-gamma}
	\end{minipage}
	\caption{(a) A simple optimal $2$-plane graph $G$; the underlying pentangulation $\sigma(G)$ of $G$ is bold. (b) A drawing of the outer face of $\sigma(G)$ (bold edges) and of its chords with two bends in total. }
\end{figure}

\begin{lemma}\label{le:2bends-in-total}
	Every optimal $2$-plane graph has a polyline drawing that fully preserves its topology and with two bends in total.
\end{lemma}
\begin{proof}
	We first prove that the planar skeleton $\sigma(G)$ of an optimal $2$-plane graph $G$ is $3$-connected. Suppose, for a contradiction, that $\sigma(G)$ contains a separation pair $\{u,v\}$ such that its removal disconnects $\sigma(G)$ into $c$ components, for some $c \ge 2$. Then there are (at least) $c$ distinct faces of $\sigma(G)$, such that both $u$ and $v$ are incident to these faces (each of these faces is shared by two of the $c$ components). Since each face of $\sigma(G)$ contains five chords in its interior, it follows that $G$ contains (at least) $c$ parallel edges having $u$ and $v$ as end-vertices, which contradicts the assumption that $G$ is simple.
	
	Let $\{a,b,c,d,e\}$ be the five vertices of the external face of $\sigma(G)$ in the order they appear when walking clockwise along its boundary. Let $P^*$ be the graph obtained by adding the chord $(b,e)$ to $\sigma(G)$. We use the algorithm by Chiba et al.~\cite{Chiba1985} to compute a  drawing $\Gamma$ of $P^*$ such that the external face is an equilateral triangle. Note that all the inner faces of $P^*$ are pentagons, except for the $4$-cycle $\{b,c,d,e\}$. Let $f$ be any pentagon of $P^*$, since it is drawn strictly convex in $\Gamma$, its five chords can be drawn with straight-line segments such that each segment is entirely contained in $f$ (except for its endpoints). Now we aim at drawing the four remaining chords of the external face of $\sigma(G)$ (the chord $(b,e)$ is already drawn) with two bends in total. Since the triangle $\{a,b,e\}$ is equilateral and the $4$-cycle $\{b,c,d,e\}$ is strictly convex, such a drawing can be obtained by representing edges $(c,e)$ and $(b,d)$ with straight-line segments and edges $(a,d)$ and $(a,c)$ with one bend each, as shown in Fig.~\ref{fi:o2p-gamma}.
\end{proof}

\begin{figure}[t]
	\centering
	\begin{minipage}[b]{.3\textwidth}
		\centering 
		\includegraphics[page=3, width=\textwidth]{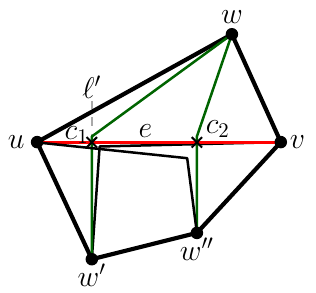}
		\subcaption{~}\label{fi:o2p-lac-1}
	\end{minipage}
	\hfil
	\begin{minipage}[b]{.3\textwidth}
		\centering 
		\includegraphics[page=4, width=\textwidth]{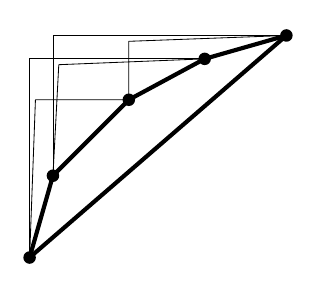}
		\subcaption{~}\label{fi:o2p-lac-2}
	\end{minipage}
	\caption{Drawing the chords of (a) an inner face and (b) the outer face, such that they use (at most) one bend each and they cross at large angles. }
\end{figure}

\begin{lemma}\label{le:LAC}
	Every optimal $2$-plane graph has a polyline drawing with curve complexity one that fully preserves its topology and such that every crossing angle is at least $\frac{\pi}{2}-\epsilon$, for any given $\epsilon > 0$.
\end{lemma}
\begin{proof}
	We use again the result by Chiba et al.~\cite{Chiba1985}. This time, we use it to compute a  drawing $\Gamma$ of $\sigma(G)$ such that the outer face is as convex polygon having all its corners along an $xy$-monotone curve; see Fig.~\ref{fi:o2p-lac-2} for an illustration. Let $f$ be an inner face of $\sigma(G)$, which is drawn strictly convex in $\Gamma$. Let $e=(u,v)$ be the longest among the chords of $f$; we shall assume that $e$ is drawn horizontal, up to a rotation of the drawing. The boundary of $f$ is formed by two paths that connect $u$ and $v$, one consisting of two edges and the other consisting of three edges. We denote by $w$ the vertex in the shorter path, and by $w'$ and $w''$ the two vertices in the other path. By possibly mirroring the drawing we can assume that $w$ is above $e$ and that $w'$ is to the left of $w''$; see Fig.~\ref{fi:o2p-lac-1}.  Consider the vertical half-line $\ell'$ starting at $w'$. Since $e$ is the longest chord of $f$, the crossing point of $\ell'$ with $e$ is inside $f$, as otherwise $(v,w')$ would be a chord of $f$ longer than $e$. Then we place the bend point on $\ell'$ slightly above $e$. The edge $(w,w'')$ is drawn analogously, so these two chords both cross $e$ orthogonally. The bend of the edge $(v,w')$ is placed below and to the right of the crossing point $c_1$ between $(w,w')$ and $e$; by choosing the bend point arbitrarily close to $c_1$ we obtain that $(v,w')$ is formed by a segment whose slope is arbitrarily close to vertical and by a segment whose slope is arbitrarily close to horizontal. Similarly, the bend of the edge $(u,w'')$ is placed below and to the left of the crossing point $c_2$ between $(w,w'')$ and $e$. In addition, the bend point of $(u,w'')$ is placed sufficiently below $e$ so that the almost horizontal segment of $(u,w'')$ intersect the almost vertical segment of $(v,w')$ and the vertical segment of $(w,w')$. A suitable choice of the bend points of $(v,w')$ and of $(u,w'')$ allow us to fix the smallest crossing angle to be arbitrarily close to $\frac{\pi}{2}$. Concerning the outer face, since all vertices are placed along an $xy$-monotone path, it is immediate to see that we can draw all its chords with one bend and such that each segment is either arbitrarily close to vertical or arbitrarily close to horizontal,  as shown in Fig.~\ref{fi:o2p-lac-2}.
\end{proof}

\begin{figure}[t]
	\centering
	\begin{minipage}[b]{.3\textwidth}
		\centering 
		\includegraphics[page=5, width=\columnwidth]{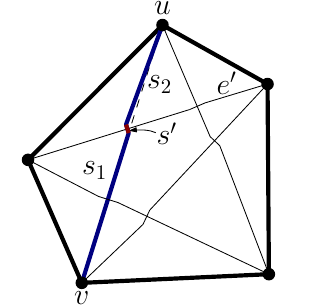}
		\subcaption{~}\label{fi:o2p-rac-1}
	\end{minipage}\hfil
	\begin{minipage}[b]{.3\textwidth}
		\centering 
		\includegraphics[page=6, width=\columnwidth]{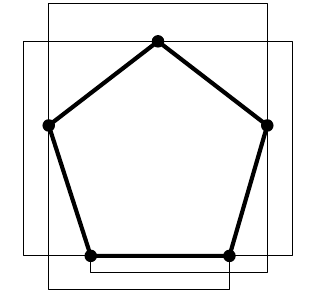}
		\subcaption{~}\label{fi:o2p-rac-2}
	\end{minipage}
	\caption{Drawing the chords of (a) an inner face and (b) the outer face, such that they use two bends each and they cross at right angles. }
\end{figure}

\begin{lemma}\label{le:RAC}
	Every optimal $2$-plane graph has a polyline drawing with curve complexity two that fully preserves its topology and such that every crossing angle is $\frac{\pi}{2}$.
\end{lemma}
\begin{proof}
	We use again the algorithm by Chiba et al.~\cite{Chiba1985}. This time, we use it to to compute a  drawing $\Gamma$ of $\sigma(G)$ such that the outer face is a regular $5$-gon. Let $f$ be an inner face of $\sigma(G)$, which is drawn strictly convex in $\Gamma$. Note that each of the five chords in $f$ is crossed twice, and hence we can assign each crossing to one of its two involved chords such that each crossing is assigned to one chord and each chord is assigned with exactly one crossing. We then draw the chords as straight-line segments and then locally modify the drawing of each chord in correspondence with the assigned crossing. Refer to Fig.~\ref{fi:o2p-rac-1}. Let $e=(u,v)$ and $e'$ be two chords that cross and let $e$ be the edge assigned with this crossing. Moreover let $s_1=\overline{up}$ and $s_2=\overline{qv}$ be the two segments obtained from $e$ by removing from it a short segment $s=\overline{pq}$ that contains the crossing point. In particular, let $s_2$ be the segment that does not contain the second crossing of the edge $e$. We replace the removed segment $s$ with another segment $s'$ having $p$ as an endpoint and such that it cross $e'$ forming a right angle. Then we slightly change the slope of $s_2$ (and its length) so that $q$ coincides with the other endpoint of $s'$. Note that the length of $s$ (and hence the length of $s'$) can be chosen sufficiently small such that no new crossings are introduced in the drawing. Also, since $s_2$ is not involved in any other crossing, we do not change the angle of any other crossings. With the same strategy all chords of $f$ can be modified to have two bends and to cross at right angles. Concerning the outer face, since it is drawn as a regular $5$-gon, it is immediate to see that we can draw all its chords with two bends and rectilinear, as in Fig.~\ref{fi:o2p-rac-2}.
\end{proof}

Lemmas~\ref{le:2bends-in-total}, \ref{le:LAC}, and \ref{le:RAC} are summarized by the next theorem.

\begin{theorem}\label{th:optimal-2-plane}
  Every optimal $2$-plane graph has a polyline drawing $\Gamma$ that fully preserves its topology and that has one of the following properties:
  
  \begin{inparaenum}[(a)]
  
    \item $\Gamma$ has two bends in total.
    
    \item $\Gamma$ has curve complexity one and every crossing angle is at least $\frac{\pi}{2}-\epsilon$, for~any~$\epsilon>0$.
    
    \item $\Gamma$ has curve complexity two and every crossing angle is exactly $\frac{\pi}{2}$.
    
  \end{inparaenum}
\end{theorem}

\section{Open Problems}\label{se:conclusions}

Theorem~\ref{th:disconnected-lower-bound} proves a lower bound of $\Omega(\sqrt{n})$ on the curve complexity of polyline drawings that partially preserve the topology and that do not have a connected skeleton. It may be worth understanding whether this bound is tight.

Theorem~\ref{th:optimal-2-plane} proves that for optimal $2$-plane graphs a crossing angle resolution arbitrarily close to $\frac{\pi}{2}$ can be achieved with curve complexity one, while optimal crossing angle of $\frac{\pi}{2}$ is achieved at the expenses of curve complexity two. Can optimal crossing angle resolution and  curve complexity one be simultaneously achieved?  A positive answer to this question is known if the planar skeleton of the graph is a dodecahedron~\cite{DBLP:conf/compgeom/Bekos0R17}.

Finally, a natural research direction suggested by the research in this paper is to extend the study of the curve complexity of drawings that fully preserve the topology to other families of beyond-planar topological graphs. For example, it would be interesting to understand whether Theorem~\ref{th:optimal-2-plane} can be extended to non-optimal 2-plane graphs.

\section*{Acknowledgements}
We wish to thank Stephen Wismath for useful discussions about the topics of this research.

%
%

\bibliography{topobends}

\end{document}